\documentclass{lmcs}
\pdfoutput=1

\usepackage{lastpage}
\lmcsdoi{18}{1}{8}
\lmcsheading{}{\pageref{LastPage}}{}{}%
{Apr.~21,~2021}{Jan.~12,~2022}{}

\usepackage[utf8]{inputenc}

\usepackage{hyperref,algorithm,algorithmic,amsmath,amssymb,enumitem}
\hypersetup{hypertexnames=false} 
\usepackage{lettrine}
\keywords{Parity Games, Quasipolynomial algorithm, Zielonka's algorithm}

\newcommand{\solve}{{\textsc{Solve}}}
\newcommand{\atr}{\textsc{Attr}}
\newcommand{\Gg}{\mathcal{G}}
\newcommand{\Even}{\mathsf{E}}
\newcommand{\Odd}{\mathsf{O}}
\newcommand{\Oo}{\mathcal{O}}

\newcommand{\player}{\wp}

\begin{document}
\title[A Recursive Quasipolynomial Solution to Parity Games]{A Recursive Approach to Solving Parity Games \texorpdfstring{\\}{} in Quasipolynomial Time\rsuper*}
\titlecomment{{\lsuper*}Journal version of Parys~\cite{Par19} and Lehtinen, Schewe, and Wojtczak~\cite{LSW19}}

\author[K.~Lehtinen]{Karoliina Lehtinen\rsuper{a}}	
\author[P.~Parys]{Pawe{\l} Parys\rsuper{b}}
\author[S.~Schewe]{Sven Schewe\rsuper{c}}	
\author[D.~Wojtczak]{Dominik Wojtczak\rsuper{c}}

\address{CNRS, Aix-Marseille University and University of Toulon, LIS, Marseille}	
\email{lehtinen@lis-lab.fr}  
\urladdr{\url{http://pageperso.lif.univ-mrs.fr/~karoliina.lehtinen/}}

\address{Institute of Informatics, University of Warsaw, Poland}
\email{parys@mimuw.edu.pl}
\urladdr{\url{https://www.mimuw.edu.pl/~parys/}}

\address{University of Liverpool, UK}	
\email{\{sven.schewe,d.wojtczak\}@liverpool.ac.uk}  
\urladdr{\url{https://www2.csc.liv.ac.uk/~sven/}, \url{https://www2.csc.liv.ac.uk/~dominik/}}

\thanks{
\lettrine[image=true, lines=2, findent=1ex, nindent=0ex, loversize=.12]{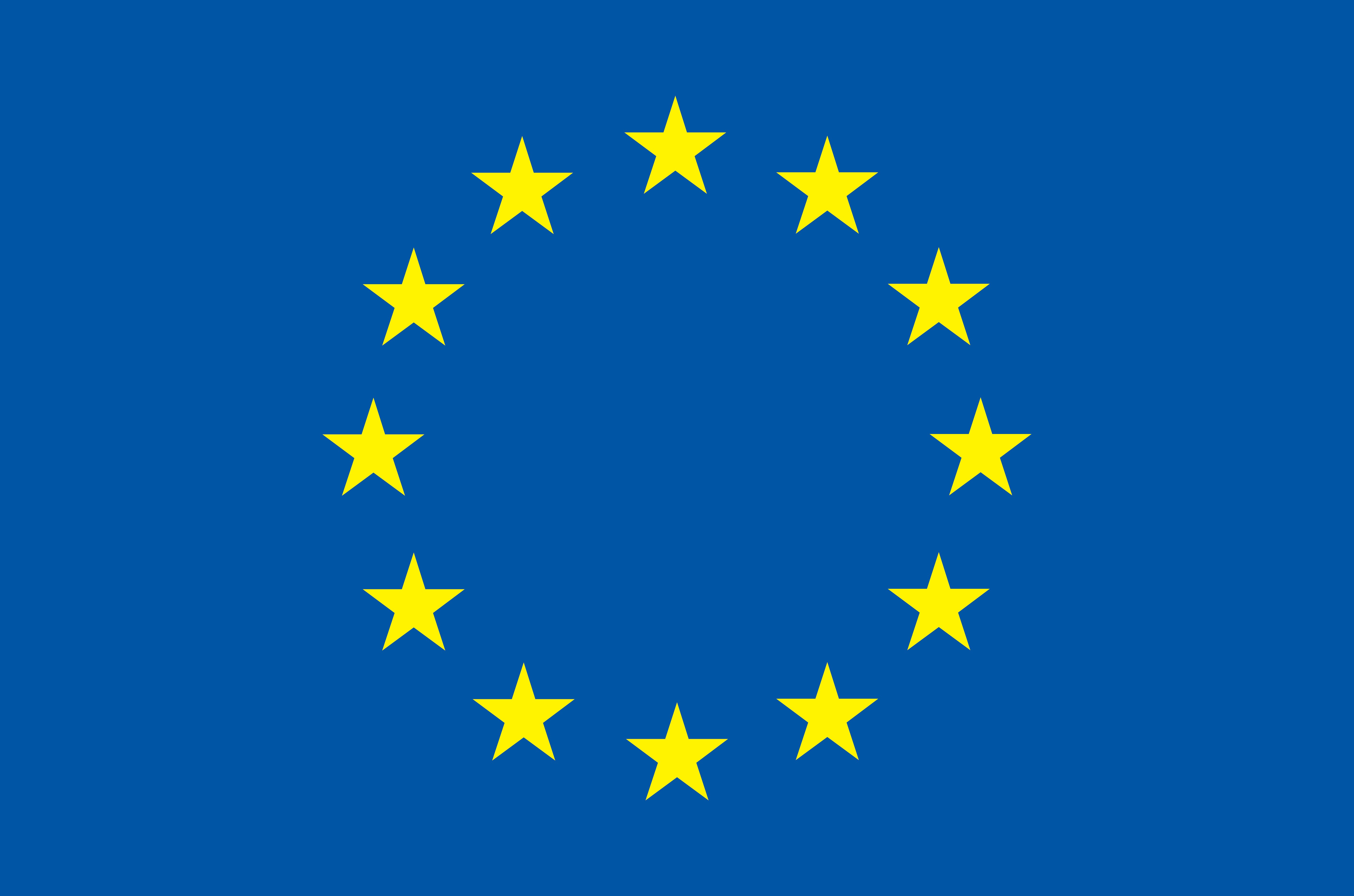}
This project has received funding from the European Union's Horizon 2020 research and innovation programme under the Marie Sk\l odowska-Curie grant
agreement No 892704, and by the by the Engineering and Physical Sciences Research Council grant EP/P020909/1. The second author is supported by the National Science Centre, Poland (grant no.\@ 2016/22/E/ST6/00041).
}	




\begin{abstract}
Zielonka's classic recursive algorithm for solving parity games is perhaps the simplest among the many existing parity game algorithms. However, its complexity is exponential, while currently the state-of-the-art algorithms have quasipolynomial complexity.
Here, we present a modification of Zielonka's classic algorithm that brings its complexity down to $n^{\Oo\left(\log\left(1+\frac{d}{\log n}\right)\right)}$, for parity games of size $n$ with $d$ priorities, in line with previous quasipolynomial-time solutions.
\end{abstract}

\maketitle
\section{Introduction}

A parity game is an infinite two-player game in which Even and her opponent Odd build an infinite path along the edges of a graph labelled with integer priorities.
Even's goal is for the highest priority seen infinitely often on this path to be even, while Odd tries to stop her.

Parity games are a central tool in automata theory, logic, and their applications to verification.
In particular, the model-checking problem for the modal $\mu$ calculus~\cite{EJS01} and the synthesis problem for LTL~\cite{PR89} reduce to solving parity games.
Solutions to parity games have also influenced work on $\omega$-word automata translations~\cite{BL18,DJL19}, linear optimisation~\cite{Fri11,FHZ11} and stochastic planning~\cite{Fea10}.

The complexity of solving parity games, that is, deciding which player has a winning strategy, is a long standing open problem.
The problem is known to be both in {\sc UP} and {\sc coUP}~\cite{Jur98} but a polynomial algorithm remains, so far, out of reach.
In 2017 Calude et al.\@ published the first quasipolynomial-time algorithm~\cite{CJKLS17}, which was followed by several alternative algorithms~\cite{JL17,Leh18},
variations and improvements thereupon~\cite{FJKSSW19,Par20,DJT20}.
The progress-measure approach~\cite{JL17} is based on succinct encodings of classical progress measures,
while the register-game approach~\cite{Leh18} is based on a relatively involved analysis of the structure of winning strategies in parity games.

In this paper we present a quasipolynomial-time parity-game algorithm based on Zielonka's classic recursive algorithm.
Zielonka's algorithm~\cite{Zielonka}, based on McNaughton's algorithm for solving Muller games~\cite{McNaughton}, is perhaps the (conceptually) simplest parity game algorithm to date.
Despite its exponential worst-case complexity, it is also one of the most performant algorithms in practice~\cite{oink}.
Here we show how to adapt this algorithm to make its worst-case complexity quasipolynomial.
Our algorithm, its correctness proof, and complexity analysis are all remarkably simple.
Its runtime complexity is roughly in line with previous quasipolynomial-time algorithms.

Our key insight is that, instead of each recursive call solving a subgame, we use a weakened induction hypothesis,
which postulates that each call should return a partition that separates the \textit{small dominions} of both players, up to a size specified by a pair of parameters.
Generalising the observation that only one dominion can be larger than half the arena,
we then use these parameterised calls to build an algorithm that only makes a quasipolynomial number of calls, but still finds the winning regions of each player.

The time complexity of our algorithm, which is in $n^{\Oo\left(\log\left(1+\frac{d}{\log n}\right)\right)}$ for games with $n$ vertices and $d$ priorities, is similar to the complexity of previous quasipolynomial-time algorithms.
This also provides fixed-parameter tractability when $d$, the number of priorities, is treated as the parameter,
as well as a polynomial bound for the common case where the number of priorities is logarithmic in the number of states.
In a fine grained comparison our algorithm (similarly to the original Zielonka's algorithms, on which it is based) operates symmetrically, going through every priority,
rather than just half of them (as in previous quasipolynomial-time algorithms),
meaning that the $\Oo\left(1+\log\frac{d}{\log n}\right)$ in the exponent hides a factor of $2$.
Thus, a very careful analysis still reveals a small gap, when compared to previous quasipolynomial-time algorithms.

We provide two versions of our algorithm.
One has better worst-case complexity (outlined above), while the other turns out to be somewhat better in practice.
We evaluate both versions against Zielonka's classic algorithm, which, despite its exponential worst-case behaviour, is in most cases faster than our quasipolynomial-time versions.
In line with the theory, both our algorithms outperform Zielonka's algorithm on the families of games designed to exhibit its exponential worst-case behaviour.

We also briefly comment on the relationship between this recursive algorithm and universal trees,
but refer the reader to Jurdzi\'nski and Morvan's work~\cite{JM20} for a more thorough analysis, which also addresses the symbolic implementation of recursive algorithms.

This is a journal version of Parys' paper~\cite{Par19}, in which he turns Zielonka's exponential-time algorithm into a quasipolynomial-time one;
additionally, it includes theoretical improvements suggested by Lehtinen, Schewe, and Wojtczak~\cite{LSW19}, which bring the complexity of the algorithm down to roughly match the state of the art.

\section{Preliminaries}

A \emph{parity game} is a two-player game between players Even and Odd
played on a \emph{game graph} defined as a tuple $\Gg=(V,V_\Even,E,\pi)$,
where $(V,E)$ is a finite directed graph in which every vertex has at least one successor;
its vertices are labelled with positive integer \emph{priorities} by $\pi\colon V\rightarrow\{1,2,\dots,d\}$ (for some $d\in\mathbb{N}$), and
partitioned between vertices $V_\Even$ \emph{belonging to Even} and vertices $V_\Odd=V\setminus V_\Even$ \emph{belonging to Odd}.
As usual for directed graphs, we forbid self-loops (i.e., edges from a vertex to itself).
We sometimes abbreviate Even and Odd to $\Even$ and $\Odd$, we use the symbol $\player$ for a player (either Even or Odd), and we write $\bar\player$ for the opponent of $\player$.

A \emph{play} $\rho$ is an infinite path through the game graph.
It is \emph{winning for Even} if the highest priority occurring infinitely often on it is even; otherwise it is \emph{winning for Odd}.
We write $\rho[i]$ for the $i$\textsuperscript{th} vertex in $\rho$ (zero based) and $\rho[0,j]$ for its prefix of length $j+1$.

A \emph{strategy} from a vertex $v$ for a player $\player$ is a function that
\begin{itemize}
\item takes any prefix of a play starting in $v$ and ending in a vertex that belongs to $\player$, and
\item returns one of successors of the latter vertex.
\end{itemize}
A strategy $\sigma$ from $v$ for Even (Odd) \emph{agrees} with a play $\rho$ if $\rho[0]=v$ and, whenever $\rho[i]\in V_\Even$ ($V_\Odd$, respectively), then $\rho[i+1]=\sigma(\rho[0,i])$
(and $\sigma$ agrees with a finite path, if it agrees with some infinite play extending the path).
A strategy for a player is \emph{winning} if it agrees only with plays winning for that player.
Parity games are determined: from every vertex, one of the two players has a winning strategy~\cite{Martin-determinacy}.

The \emph{winning region} of a player $\player$ is the set of vertices from which $\player$ has a winning strategy.
We are interested in the problem of computing, given a parity game $G$, the winning regions of each player.

Given a set of vertices $G\subseteq V$, we can consider a \emph{subgame} of $\Gg$ \emph{induced} by $G$,
which is obtained simply by restricting all components of $\Gg$ to $G$.
Notice, however, that not every set $G$ induces a valid subgame: every vertex in $G$ should have at least one successor in $G$.
In the sequel, we often write ``the subgame $G$'' instead of the more precise ``the subgame of $\Gg$ induced by $G$''.
For the definitions below, we assume some set $G\subseteq V$ inducing a subgame of $\Gg$.

A \emph{dominion} of a player $\player$ in a subgame $G$ is a set of vertices $D\subseteq G$ such that
from every vertex $v\in D$ the player $\player$ has, in the subgame $G$, a winning strategy that agrees only with plays staying forever in $D$.

Given a set $S\subseteq G$, the \emph{$\player$-attractor} of $S$ in the subgame $G$, written $\atr_\player(S,G)$,
is the set of vertices from which the player $\player$ has a strategy in $G$ that agrees only with plays reaching $S$.
A set $S\subseteq G$ is \emph{$\player$-closed} in the subgame $G$ if $\atr_{\bar \player}(G\setminus S,G)=G\setminus S$.
More concretely, this means that $\player$ can ensure to stay inside $S$: every vertex of $\player$ in $S$ has at least one successor in $S$,
and every vertex of $\bar\player$ in $S$ has no successor in $G\setminus S$.

Observe that the set $G\setminus\atr_\player(S,G)$, for any $S\subseteq G$, always induces a subgame:
if all successors of a vertex belong to $\atr_\player(S,G)$, then this vertex belongs there as well.
Observe also that every dominion of $\player$ is $\player$-closed, and that the winning region of $\player$ is a dominion of $\player$.

We use the following simple lemmata to prove correctness of our algorithm:

\begin{lem}\label{lem:att}
	Let $D$ be a dominion of a player $\player$ in a subgame $G$, and let $X\subseteq G$.
	If $D$ does not intersect with $X$, then $D$ is also a dominion of $\player$ in {\normalfont$G\setminus\atr_{\bar \player}(X,G)$}.
\end{lem}

\begin{proof}
	By the definition of a dominion, from every vertex $v\in D$ player $\player$ has a winning strategy that agrees only with plays staying in $D$.
	Observe that a play agreeing with such a strategy cannot reach a vertex in the $\bar \player$-attractor of $X$,
	because from such a vertex the opponent can force to leave $D$ and enter $X$.
	Thus, the same strategies (after restricting them appropriately) witness that $D$ is also a dominion in the smaller subgame;
	in particular $D$ does not intersect with $\atr_{\bar \player}(X,G)$.
\end{proof}

\begin{lem}\label{lem:domstrat}
	Let $D$ be a dominion of a player $\player$ in a subgame $G$, and let $X\subseteq G$ be $\bar\player$-closed in $G$.
	Then $D\cap X$ is a dominion of $\player$ in $X$.
\end{lem}

\begin{proof}
	Because $X$ is $\bar\player$-closed in $G$, no vertex of $\player$ in $D\cap X$ has a successor in $G\setminus X$.
	In consequence, strategies that witness $D$ being a dominion of $\player$ in $G$, after restricting them appropriately,
	witness $D\cap X$ being a dominion of $\player$ in $X$.
\end{proof}

\begin{lem}\label{lem:no-h}
	If all priorities in a non-empty dominion $D$ of a player $\player$ are at most $d'$, and $d'$ is not of $\player$'s parity,
	then $D$ contains a non-empty sub-dominion $C$ without vertices of priority $d'$.
\end{lem}

\begin{proof}
	Fix a strategy $\sigma$ for $\player$ from some vertex $v_0\in D$ that is winning and agrees only with plays staying in $D$.
	To $C$ we take all vertices $v$ of $D$ such that
	\begin{itemize}
	\item	there exists a finite path $\rho_v$ that ends in $v$ and agrees with $\sigma$, and
	\item	there is no finite path $\rho'$ that extends $\rho_v$, agrees with $\sigma$, and ends in a vertex of priority $d'$ (in particular, $v$ itself is not of priority $d'$).
	\end{itemize}
	Such a set $C$ is nonempty; otherwise, it would be easy to construct a play that agrees with $\sigma$ and visits priority $d'$ infinitely often
	(start from $v_0$, continue to a vertex of priority $d'$, which is possible because $v_0\not\in C$, continue to an arbitrary successor $v_1$,
	continue to a vertex of priority $d'$, which is possible because $v_1\not\in C$, and so on).
	Then, for a vertex $v\in C$ we consider a strategy $\sigma_v$ such that $\sigma_v(\rho)=\sigma(\rho_v\cdot\rho)$ for all plays $\rho$ starting in $v$;
	this strategy is winning for $\player$ (if $\sigma_v$ agrees with a play $\rho$, then $\sigma$ agrees with $\rho_v\cdot\rho$, and hence both these plays are winning)
	and agrees only with plays staying in $C$ (by the definition of $C$).
\end{proof}

\section{A first quasipolynomial-time version}\label{sec:liverpool}

In this section we present our first quasipolynomial-time algorithm solving parity games.
Informally, we refer to it as the Liverpool variant,
since it is (close to) an algorithm introduced by Lehtinen, Schewe, and Wojtczak~\cite{LSW19}, working at University of Liverpool.

We remark that, historically, the Liverpool variant came after the Warsaw variant, presented in the next section.
We present the Liverpool variant first, as the ``main variant'', because it has better theoretical complexity and because we think that it is more elegant.

\subsection{The algorithm}

First, recall Zielonka's classic recursive algorithm.
Each iteration first identifies the set $N_d$ of vertices of the highest priority, which we can assume to  be even as the odd case is symmetric, and the $\Even$-attractor of $N_d$.
It then makes a recursive call to solve the rest of the game, which has one fewer priority.
Odd's winning region in this subgame, as well as its $\Odd$-attractor, is winning for Odd in the whole game, and is considered solved.
The algorithm then iterates on the remaining arena, which is smaller than the original arena, until no more winning region for Odd is found, and the rest of the arena can be declared winning for Even.
The depth of the recursion is bounded by the number of priorities, while the number of iterations of the while loop is bounded by the size of the arena.
This yields an exponential complexity for the algorithm of roughly $O(n^d)$.

To make this algorithm quasipolynomial-time, we weaken and parameterise the inductive guarantees of recursive calls.
Instead of expecting a recursive call to solve a subgame, we ask that it provides a partition of the subgame into two regions,
one which contains all of Odd's dominions of size up to a given parameter, and one which contains all of Even's dominions of size up to another parameter.
These two parameters determine the precision of the recursive call; if they are both the size of the whole subgame, this corresponds to solving it entirely.

Using three recursive calls, our algorithm consecutively computes three regions, $G\setminus G_1$, $G_1\setminus G_2$, and $G_2\setminus G_3$, which, together,
contain all Odd's dominions of size up to the first parameter and no Even's dominion of size up to the second parameter.
The first and the third region are based on calls that use half the precision parameter for Odd's dominions,
and only the second region is computed using the full precision.
Correctness of the algorithm hinges on proving that $G\setminus G_1$ and $G_1\setminus G_2$ together cover already over the half of any small Odd's dominion,
and hence the last call handles correctly what is left.

We present $\solve_\Even(G,d,p_\Even,p_\Odd)$, which, given a subgame $G$ with priorities up to an even number $d$,
returns a set of vertices containing all of Even's small dominions (of size up to $p_\Even$) and not intersecting any of Odd's small dominions (of size up to $p_\Odd$).
The dual, $\solve_\Odd$, is defined for odd $d$ by swapping $\Even$ with $\Odd$.

\begin{algorithm}
\caption{$\solve_\Even(G,d,p_\Even,p_\Odd)$}\label{alg:simple}
\begin{algorithmic}[1]
	\STATE \algorithmicif\ $G = \emptyset$ or $p_\Odd\leq 1$\ \algorithmicthen
	\STATE \hskip\algorithmicindent\algorithmicreturn\ $G$\label{l2}
	\STATE $G_1=\solve_\Even(G,d, p_\Even, \lfloor p_\Odd/2\rfloor)$\label{l4}
	\STATE $N_d := \{v\in G_1 \mid \pi(v)=d\}$
	\STATE $H := G_1 \setminus \atr_\Even(N_d,G_1)$\label{l6}
	\STATE $W_\Odd := \solve_\Odd(H,d-1,p_\Odd,p_\Even)$\label{l7}
	\STATE $G_2 := G_1 \setminus \atr_\Odd(W_\Odd,G_1)$\label{l8}
	\STATE $G_3:=\solve_\Even(G_2,d,p_\Even,\lfloor p_\Odd/2\rfloor)$\label{l9}
	\STATE \algorithmicreturn\ $G_3$
\end{algorithmic}
\end{algorithm}

\begin{rem}
	In Lehtinen, Schewe, and Wojtczak~\cite{LSW19}, Line~\ref{l4} of the above algorithm was replaced by
	$G_1=G\setminus\atr_\Odd(G\setminus\solve_\Even(G,d, p_\Even, \lfloor p_\Odd/2\rfloor), G)$.
	We can see, however, that the call to $\atr_\Odd$ in this expression adds no new vertices, meaning that such a line computes exactly the same set $G_1$ as our Line~\ref{l4}.
	Indeed, below, in Lemma~\ref{lem:cor-liv}, we show that the set returned by $\solve_\Even(G,d, p_\Even, \lfloor p_\Odd/2\rfloor)$ is $\Even$-closed.
	Thus, on the one hand, the presence of $\atr_\Odd$ is redundant here; on the other hand, after adding this $\atr_\Odd$ one can simplify a bit the correctness proof:
	Item~\ref{c3:e1} of Lemma~\ref{lem:cor-liv} ceases to be needed.
\end{rem}

\subsection{Correctness}

We prove the following lemma, which guarantees that $\solve_\Even(G,d,p_\Even,\allowbreak p_\Odd)$ and $\solve_\Odd(G,d,p_\Odd,p_\Even)$
partition a subgame $G$ of maximal priority at most $d$ into a region that contains all Odd's dominions of size up to $p_\Odd$ and a region that contains all Even's dominions of size up to $p_\Even$.
Then $\solve_\Even(G,d,|G|,|G|)$ (if $d$ is even) or $\solve_\Odd(G,d,|G|,|G|)$ (if $d$ is odd) solves $G$ completely.

\begin{lem}\label{lem:cor-liv}
	The procedure {\normalfont$\solve_\Even(G,d,p_\Even,p_\Odd)$}, where $d$ is even and not smaller than the maximal priority in $G$, returns a set that
	\begin{enumerate}[label={(\roman*)}]
	\item\label{c3:e1}
		is $\Even$-closed in $G$,
	\item\label{c3:e2}
		contains all Even's dominions in $G$ of size up to $p_\Even$, and
	\item\label{c3:e3}
		does not intersect with any Odd's dominion in $G$ of size up to $p_\Odd$.
	\end{enumerate}
	Similarly,  {\normalfont$\solve_\Odd(G,d,p_\Odd,p_\Even)$}, where $d$ is odd and not smaller than the maximal priority in $G$, returns a set that
	\begin{enumerate}[label={(\roman*)}]
	\item	is $\Odd$-closed in $G$,
	\item	contains all Odd's dominions in $G$ of size up to $p_\Odd$, and
	\item	does not intersect with any Even's dominion in $G$ of size up to $p_\Even$.
	\end{enumerate}
\end{lem}

\begin{proof}
	The proof is by induction on the sum $d+p_\Even+p_\Odd$.
	We present only the part concerning $\solve_\Even(G,d,p_\Even,p_\Odd)$, as the other part is similar.

	When $G=\emptyset$ or $p_\Odd\leq 1$, the whole $G$ returned in Line~\ref{l2} clearly satisfies the thesis:
	there exist no Odd's dominions of size $1$.
	We proceed with $G\neq\emptyset$ and $p_\Odd>1$.

	We first show Item~\ref{c3:e1}.
	By the induction hypotheses we know that $G_1$ is $\Even$-closed in $G$ (which, in particular, allows us to use $G_1$ as a subgame in Line~\ref{l6}) and that $G_3$ is $\Even$-closed in~$G_2$.
	Moreover, $G_2$ is $\Even$-closed in $G_1$, because
	$G_1\setminus G_2 = \atr_\Odd(W_\Odd,G_1)$ is invariant under the $\atr_\Odd$ operation.
	In consequence the returned set $G_3$ is $\Even$-closed in $G$ as needed. This is because
	every Even's vertex in $G_3$ has at least one successor in $G_3$ since $G_3$ is $\Even$-closed in $G_2$;
	at the same time, every Odd's vertex in $G_3\subseteq G_2\subseteq G_1$ has no successor in $G\setminus G_1$ (as $G_1$ is $\Even$-closed in $G$),
	nor in $G_1\setminus G_2$ (as $G_2$ is $\Even$-closed in $G_1$),
	nor in $G_2\setminus G_3$ (as $G_3$ is $\Even$-closed in $G_2$).

	Next, we show Item~\ref{c3:e2}, saying that the set returned by $\solve_\Even(G,d,p_\Even,p_\Odd)$ contains all Even's dominions of size up to $p_\Even$.
	Let $D$ be such a dominion.
	According to the induction hypothesis, $D$ is contained in $G_1$, and it is, moreover, an Even's dominion in $G_1=G\setminus\atr_\Odd(G\setminus G_1,G)$ (cf.\@ Lemma~\ref{lem:att}).
	Since $H$ is $\Odd$-closed in $G_1$, the intersection $D'$ of $D$ and $H$
	is an Even's dominion in $H$ (cf.\@ Lemma~\ref{lem:domstrat}) and therefore, from the induction hypothesis,
	neither $D'$ nor $D$ (containing additionally only vertices not in $H$) intersects with $W_\Odd$.
	In consequence, $D$ is a dominion in $G_2$ (cf.\@ Lemma~\ref{lem:att})
	and, by the induction hypothesis, it is contained in the returned set $G_3$.

	We proceed with showing Item~\ref{c3:e3}, saying that the set returned by $\solve_\Even(G,d,p_\Even,p_\Odd)$ does not intersect with Odd's dominions of size up to $p_\Odd$.
	Let $D$ be such a dominion, let $S$ be the union of Odd's dominions of size up to $\lfloor p_\Odd/2 \rfloor$ in the subgame induced by $D$,
	and let $A$ be the $\Odd$-attractor of $S$ in this subgame.

	Because $D$ is $\Odd$-closed in $G$ (i.e., Even cannot force to leave $D$), $A$ is contained in the $\Odd$-attractor of $S$ in the whole $G$.
	For the same reason, every Odd's dominion in the subgame $D$ is an Odd's dominion in the whole $G$.
	Thus, by the induction hypothesis, $S$ does not intersect with $G_1$.
	Since (again by the induction hypothesis) $G_1$ is $\Even$-closed in $G$, this implies that $\atr_\Odd(S,G)$ (hence also its subset $A$) does not intersect with $G_1$:
	because of $S\subseteq G\setminus G_1$ we have that $\atr_\Odd(S,G)\subseteq\atr_\Odd(G\setminus G_1,G)=G\setminus G_1$.
	If $A=D$, then $D$ does not intersect with $G_3\subseteq G_1$, and we are done.

	We consider the case of $A\neq D$.
	Notice that $D$ is a dominion not only in the whole $G$, but also in the subgame induced by $D$.
	In consequence, $D\setminus A$ is a dominion in the subgame induced by $D\setminus A$ (cf.\@ Lemma~\ref{lem:domstrat}; $D\setminus A$ is $\Even$-closed in $D$).
	It contains a nonempty Odd's dominion $C$ without vertices of priority $d$ (cf.\@ Lemma~\ref{lem:no-h}).
	The union $C\cup A$ is an Odd's dominion in $D$: inside $C$ Odd has a winning strategy, which is valid as long as the play does not leave $D\setminus A$ entering $A$;
	in $A\setminus S$ Odd uses his strategy to reach $S$, and in $S$ he uses his strategy from one of the small dominions covering $S$.
	Moreover, because $D$ is $\Odd$-closed in $G$, $C\cup A$ is an Odd's dominion also in the whole $G$.
	Since it is not contained in $S$, it is of size greater than $\lfloor p_\Odd/2\rfloor$.
	We now show that $C\cup A$ does not intersect with $G_2$.

	Recall that $A\cap G_1=\emptyset$.
	Since $G_1$ is $\Even$-closed in $G$, the set $(C\cup A)\cap G_1$ is an Odd's dominion in $G_1$ (cf.\@ Lemma~\ref{lem:domstrat}),
	and since $(C\cup A)\cap G_1\subseteq C$ contains no vertices of priority $d$, also in $H=G_1\setminus\atr_\Even(N_d,G_1)$ (cf.\@ Lemma~\ref{lem:att}).
	By the induction hypothesis, it is contained in $W_\Odd$, and therefore $C\cup A$ does not intersect with $G_2\subseteq G_1\setminus W_\Odd$.

	Note that the proof of Item~\ref{c3:e1} above also shows that $G_2$ is $\Even$-closed in $G$.
	Now, since $D$ is a dominion in $G$, Lemma~\ref{lem:domstrat} implies that the set $D\cap G_2$ is a dominion in $G_2$.
	Moreover, its size is at most $p_\Odd-(\lfloor p_\Odd/2\rfloor+1)\leq \lfloor p_\Odd/2 \rfloor$,
	because $D$ is of size at most $p_\Odd$ and $A\cup C \subseteq D$ (not intersecting with $G_2$) is of size greater than $\lfloor p_\Odd/2\rfloor$.
	So, by the induction hypothesis, $D\cap G_2$ does not intersect with the returned set $G_3$.
	The same holds for the whole $D$, because $G_3 \subseteq G_2$.
\end{proof}

\subsection{Analysis}

Let $R(d,\ell)$ be the maximal number of calls to $\solve_\Even$ and $\solve_\Odd$ performed during an execution of $\solve_\Even(G,d,p_\Even,p_\Odd)$ if $d$ is even
or of $\solve_\Odd(G,d,p_\Odd,p_\Even)$ if $d$ is odd, where $\ell=\lfloor \log p_\Even \rfloor + \lfloor \log p_\Odd \rfloor$ (in this paper, $\log$ denotes the binary logarithm).
We assume that $p_\Even\geq 1$ and $p_\Odd\geq 1$.

An induction on $\ell+d$ shows that $R(d,\ell)\leq 2^{\ell+1}\cdot\binom{d+\ell}{\ell}-1$.
Indeed, if $d=0$ (implying $G=\emptyset$) or $\ell=0$ (implying $p_\Even=p_\Odd=1$), then the execution finishes immediately,
so $R(d,\ell)=1\leq 2^{\ell+1}\cdot\binom{d+\ell}{\ell}-1$.
For positive $d$ and $\ell$ we have
\begin{align*}
	R(d,\ell) & \leq 1+2\cdot R(d,\ell-1) + R(d-1,\ell),
\end{align*}
where $1$ counts the main call itself, $2\cdot R(d,\ell-1)$ counts calls performed while executing Lines~\ref{l4} and~\ref{l9},
and $R(d-1,\ell)$ counts calls performed while executing Line~\ref{l7}.
Then, applying the induction hypothesis, we obtain that
\begin{align*}
	R(d,\ell)& \leq 1+2\cdot\left(2^{\ell-1+1}\cdot\binom{d+\ell-1}{\ell-1}-1\right) + 2^{\ell+1}\cdot\binom{d-1+\ell}{\ell}-1 \\
	& = 2^{\ell+1}\cdot\binom{d+\ell}{\ell}-2
	\leq 2^{\ell+1}\cdot\binom{d+\ell}{\ell}-1.
\end{align*}

We now apply the inequality $\binom{k}{\ell}\leq{\left(\frac{ek}{\ell}\right)}^\ell$ for $k=d+\ell$, obtaining
\begin{align*}
	R(d,\ell) & \leq 2^{\ell+1}\cdot{\left(\frac{e\cdot(d+\ell)}{\ell}\right)}^\ell
	=2\cdot 2^{\ell\left(1+\log e + \log \left(1+\frac{d}{\ell}\right)\right)}.
\end{align*}

Recall that while solving a parity game with $n$ vertices, we start with $p_\Even=p_\Odd=n$, that is, with $\ell=2\cdot\lfloor\log n\rfloor$.
One can check that the above function is increasing, meaning that we can replace $\ell$ by $2\cdot\log n$.
We obtain that
\begin{align*}
	R(d,2\cdot\lfloor\log n\rfloor)&\leq 2\cdot 2^{2\cdot(\log n)\cdot\left(1+\log e + \log \left(1+\frac{d}{2\cdot\log n}\right)\right)}
	=2\cdot n^{2\cdot\left(1+\log e + \log \left(1+\frac{d}{2\cdot\log n}\right)\right)}.
\end{align*}
Observe that the cost of each call, excluding the cost of further calls performed recursively, is dominated by the attractor construction in Lines~\ref{l6} and~\ref{l8},
which is linear in the number of edges in the game.
We can also estimate $\log e\leq 1.45$.

Regarding the space complexity, we notice that the recursion has depth $\Oo(\log n)$, and at each level of recursion we only need to store some sets of vertices of size $\Oo(n)$.

The complexity is summarised in the following theorem:

\begin{thm}
	Algorithm~\ref{alg:simple} computes winning regions in a given parity game.
	For a parity game with $n$ vertices, $m$ edges, and maximal priority $d$, its memory usage is in $\Oo(m+n\cdot\log n)$, and its running time is in
	\begin{align*}
		\Oo\!\left(m\cdot n^{4.9 + 2\cdot\log \left(1+\frac{d}{2\cdot\log n}\right)}\right)=n^{\Oo\left(\log \left(1+\frac{d}{\log n}\right)\right)}.
	\end{align*}
\end{thm}

The complexity result can be reformulated as follows:
\begin{itemize}
\item	When $d=o(\log n)$, the component $\frac{d}{2\cdot\log n}$ converges to $0$, meaning that the complexity becomes $\Oo\!\left(m\cdot n^{4.9}\right)$
	(for $n$ large enough we have that $\log\left(1+\frac{d}{2\cdot\log n}\right)\leq 1.45-\log e$).
\item	This implies that our algorithm is an fpt-algorithm when $d$ is treated as a parameter.
\item	When $d=\Oo(\log n)$, the component $\frac{d}{2\cdot\log n}$ is in $\Oo(1)$, so the complexity of the algorithm is polynomial
	(with an exponent depending on the precise relation between $d$ and $\log n$).
\item	When $d=\omega(\log n)$, the component $\frac{d}{2\cdot\log n}$ dominates the component $1$,
	implying that $\log\left(1+\frac{d}{2\cdot\log n}\right)=\log\frac{d}{\log n}-1+o(1)$; in consequence, the complexity can be written as
	$\Oo\left(m\cdot n^{2.9 + 2\cdot\log \frac{d}{\log n}}\right)$ or $n^{\Oo\left(\log \frac{d}{\log n}\right)}$.
\end{itemize}

\section{A variation}

We now present Algorithm~\ref{alg:parys}, the original algorithm of Parys~\cite{Par19}, as a variation, which we refer to as the Warsaw version.

In Algorithm~\ref{alg:simple} there is one ``central'' call to $\solve_\Odd$ for a subgame with less priorities, with unchanged precision parameters.
This call is surrounded by two calls to $\solve_\Even$, where we do not decrease the number of priorities, but we divide the precision parameter $p_\Odd$ by~$2$.
While executing each of these calls, we perform one call to $\solve_\Odd$ for a subgame with less priorities, with $p_\Odd$ divided by $2$,
which is now surrounded by two calls to $\solve_\Even$ for $p_\Odd$ divided by $4$.
Thus, while looking at calls to $\solve_\Odd$ for subgames with less priorities, we see that there is one central call with unchanged $p_\Odd$;
it is surrounded by two call with $p_\Odd$ divided by $2$;
each of them is surrounded by two call with $p_\Odd$ divided by $4$, surrounded, in turn, by calls with $p_\Odd$ divided by $8$, and so on
(this ends when the result of the division becomes smaller than $2$).

In the variation, which we now present, this structure of recursive calls to $\solve_\Odd$ is modified:
we again perform one call with unchanged $p_\Odd$, but this time all the surrouning recursive calls use $p_\Odd$ divided by $2$ (not by powers of $2$).

\subsection{The algorithm}

As previously, the procedure $\solve_\Even(G,d,p_\Even,p_\Odd)$, given a subgame $G$ with priorities up to an even number $d$,
returns a set of vertices containing all of Even's small dominions (of size up to $p_\Even$) and not intersecting any of Odd's small dominions (of size up to $p_\Odd$).
The dual, $\solve_\Odd$, is defined for odd $d$ by swapping $\Even$ with $\Odd$.

\begin{algorithm}[H]
\caption{$\solve_\Even(G,d,p_\Even,p_\Odd)$}
\begin{algorithmic}[1]\label{alg:parys}
	\STATE \algorithmicif\ $G = \emptyset$ or $p_\Even\leq 1$\ \algorithmicthen\label{lp1}
	\STATE \hskip\algorithmicindent\algorithmicreturn\ $\emptyset$\label{lp5}
	\REPEAT
		\STATE $N_d := \{v\in G \mid \pi(v)=d\}$\label{lp8}
		\STATE $H := G \setminus \atr_\Even(N_d,G)$\label{lp9}
		\STATE $W_\Odd := \solve_\Odd(H,d-1,\lfloor p_\Odd/2\rfloor,p_\Even)$\label{lp10}
		\STATE $G := G \setminus \atr_\Odd(W_\Odd,G)$\label{lp11}
	\UNTIL{$W_\Odd = \emptyset$}
	\STATE $W_\Odd := \solve_\Odd(H,d-1,p_\Odd,p_\Even)$\label{lp16}
	\STATE $G := G \setminus \atr_\Odd(W_\Odd,G)$\label{lp17}
	\WHILE{$W_\Odd\neq \emptyset$}
		\STATE $N_d := \{v\in G \mid \pi(v)=d\}$\label{lp19}
		\STATE $H := G \setminus \atr_\Even(N_d,G)$
		\STATE $W_\Odd := \solve_\Odd(H,d-1,\lfloor p_\Odd/2\rfloor,p_\Even)$\label{lp21}
		\STATE $G := G \setminus \atr_\Odd(W_\Odd,G)$\label{lp22}
	\ENDWHILE\label{lp23}
	\STATE \algorithmicreturn\ $G$\label{lp24}
\end{algorithmic}
\end{algorithm}

The variation is closer to the classic recursive algorithms, using only calls to games with less priorities.
As a consequence, the number of calls is not capped (beyond the trivial cap of up to $n$ calls for games with $n$ vertices),
but all calls refer to sub-games with reduced maximal priority, while all but one asks for lesser guarantees.

\subsection{Correctness}

We prove the following lemma, which guarantees that $\solve_\Even(G,d,p_\Even,\allowbreak p_\Odd)$ and $\solve_\Odd(G,d,p_\Odd,p_\Even)$,
for a subgame $G$ of maximal priority at most $d$, returns two regions that contain all Even's dominions of size up to $p_\Even$ and all Odd's dominions of size up to $p_\Odd$, respectively.
Then $\solve_\Even(G,d,|G|,|G|)$ (if $d$ is even) or $\solve_\Odd(G,d,|G|,|G|)$ (if $d$ is odd) solves $G$ completely.

We remark that, although the correctness proofs of Algorithms~\ref{alg:simple} and~\ref{alg:parys} are quite similar,
there is also a lot of differences in details.
For readability reasons, we have choosen to leave the two proofs independent.

\begin{lem}\label{lem:cor-waw}
	The procedure {\normalfont$\solve_\Even(G,d,p_\Even,p_\Odd)$}, where $d$ is even and not smaller than the maximal priority in $G$, returns a set that
	\begin{enumerate}[label={(\roman*)}]
	\item\label{c4:e1}
		contains all Even's dominions in $G$ of size up to $p_\Even$, and
	\item\label{c4:e2}
		does not intersect with any Odd's dominion in $G$ of size up to $p_\Odd$.
	\end{enumerate}
	Similarly, {\normalfont$\solve_\Odd(G,d,p_\Odd,p_\Even)$}, where $d$ is odd and not smaller than the maximal priority in $G$, returns a set that
	\begin{enumerate}[label={(\roman*)}]
	\item	contains all Odd's dominions in $G$ of size up to $p_\Odd$, and
	\item	does not intersect with any Even's dominion in $G$ of size up to $p_\Even$.
	\end{enumerate}
\end{lem}

\begin{proof}
	The proof is by induction on the sum $d+p_\Even+p_\Odd$.
	We present only the part concerning $\solve_\Even(G,d,p_\Even,p_\Odd)$, as the other part is similar.

	When $G=\emptyset$ or $p_\Even\leq 1$, the empty set returned in Line~\ref{lp5} clearly satisfies the thesis: there exist no Even's dominions of size $1$.
	We proceed with $G\neq\emptyset$ and $p_\Even>1$.

	Before starting, we remark that the loops terminate, as $G$ shrinks in size in all iterations but the last one of each loop.

	We first show Item~\ref{c4:e1}, saying that the set returned by $\solve_\Even(G,d,p_\Even,p_\Odd)$ contains all Even's dominions of size up to $p_\Even$.
	Let $D$ be such a dominion.

	For each repetition of the body of the \textbf{repeat} (Lines~\ref{lp8}--\ref{lp11}) and \textbf{while} (Lines~\ref{lp19}--\ref{lp22}) loops, as well as for Lines~\ref{lp16}--\ref{lp17},
	we have the following inductive argument.
	Initially, $D$ is an Even's dominion in $G$.
	Since $H$ is $\Odd$-closed in $G$, the intersection $D'$ of $D$ and $H$ is an Even's dominion in $H$ (cf.\@ Lemma~\ref{lem:domstrat}) and therefore, from the induction hypothesis,
	neither $D'$ nor $D$ (containing additionally only vertices not in $H$) intersects with the set $W_\Odd$ (computed in Line~\ref{lp10}/\ref{lp16}/\ref{lp21}).
	In consequence, $D$ is a dominion in the newly computed set $G$ in Line~\ref{lp11}/\ref{lp17}/\ref{lp22} (cf.\@ Lemma~\ref{lem:att}).

	From this argument it follows that $D$ is contained in the set $G$ when it is returned in Line~\ref{lp24}.

	We proceed with showing Item~\ref{c4:e2}, saying that the set returned by $\solve_\Even(G,d,p_\Even,p_\Odd)$ does not intersect with Odd's dominions of size up to $p_\Odd$.
	Let $D$ be such a dominion.

	For each repetition of the body of the \textbf{repeat} (Lines~\ref{lp8}--\ref{lp11}) and \textbf{while} (Lines~\ref{lp19}--\ref{lp22}) loops, as well as for Lines~\ref{lp16}--\ref{lp17},
	we have the following inductive argument.
	Initially, $D \cap G$ is an Odd's dominion in $G$.
	Due to Lemma~\ref{lem:domstrat}, $D\cap G$ is then again an odd dominion in $G$ after execution of Line~\ref{lp11}/\ref{lp17}/\ref{lp22}
	(independent of how the respective set $W_\Odd$ is calculated).

	We now assume for contradiction that the intersection of $D$ and $G$ is non-empty in Line~\ref{lp24}, and
	therefore that this intersection is non-empty throughout the whole execution of the procedure.

	To establish the contradiction, we first observe that, in the last iteration of the \textbf{repeat} loop, $W_\Odd$ is empty.
	By Lemma~\ref{lem:no-h}, the intersection of $D$ and $G$, being non-empty, contains a nonempty Odd's dominion $C$ in $G$ without vertices of priority $d$
	and, by Lemma~\ref{lem:att}, $C$ is an Odd's dominion in the subgame $H=G\setminus\atr_\Even(N_d,G)$ computed in Line~\ref{lp9}.
	As $W_\Odd$ is empty, $C$ (as well as its superset $H$) must be of size greater than $\lfloor p_\Odd/2\rfloor$,
	as $C$ would otherwise be contained in the empty set $W_\Odd$ by the induction hypothesis.

	As the set $H$ from Line~\ref{lp9} is re-used in Line~\ref{lp16},
	it therefore contains the same Odd's dominion $C$.
	By the induction hypothesis, $C$ is also contained in the region $W_\Odd$ computed in Line~\ref{lp16}.
	The intersection of $D$ and the remaining game $G$ produced in Line~\ref{lp17} is of size at most $p_\Odd-(\lfloor p_\Odd/2\rfloor+1)\leq \lfloor p_\Odd/2 \rfloor$,
	because all vertices in $C$ are removed from $G$ at that point.

	Then, the same argument as for the \textbf{repeat} loop ensures that when the \textbf{while} loop finishes,
	the intersection of $D$ and $G$ contains an Odd's dominion $C'$ of size greater than $\lfloor p_\Odd/2\rfloor$.
	This leads to a contradiction with the upper bound on the size of $D\cap G$ established in the previous paragraph.
\end{proof}

\subsection{Analysis}

As in the previous section, let $R(d,\ell)$ be the maximal number of calls to $\solve_\Even$ and $\solve_\Odd$ performed during an execution of $\solve_\Even(G,d,p_\Even,p_\Odd)$ if $d$ is even
or of $\solve_\Odd(G,d,p_\Odd,p_\Even)$ if $d$ is odd, where $\ell=\lfloor \log p_\Even \rfloor + \lfloor \log p_\Odd \rfloor$.
We assume that $p_\Even\geq 1$ and $p_\Odd\geq 1$.

An induction on $\ell+d$ shows that $R(d,\ell)\leq 2\cdot n^\ell\cdot\binom{d+\ell}{\ell}-1$,
where $n\geq 1$ is the number of vertices in the considered game $\Gg$.
Indeed, if $d=0$ (implying $G=\emptyset$) or $\ell=0$ (implying $p_\Even=p_\Odd=1$), then the execution finishes immediately,
so $R(d,\ell)=1\leq 2\cdot n^\ell\cdot\binom{d+\ell}{\ell}-1$.
For positive $d$ and $\ell$ we have
\begin{align*}
	R(d,\ell) & \leq 1+n\cdot R(d-1,\ell-1) + R(d-1,\ell),
\end{align*}
where $1$ counts the main call itself, $n\cdot R(d-1,\ell-1)$ counts calls performed while executing Lines~\ref{lp10} and~\ref{lp21},
and $R(d-1,\ell)$ counts calls performed while executing Line~\ref{lp16}.
The calls with reduced parameters (Lines~\ref{lp10} and~\ref{lp21}) are always done with games of decreasing size, and there is no game of size $1$;
thus, there is at most $n$ of these calls (for games of size $n,n-1,\dots,2$ and then $0$).
Then, applying the induction hypothesis to the above formula, we obtain that
\begin{align*}
	R(d,\ell)& \leq 1+n\cdot\left(2\cdot n^{\ell-1}\cdot\binom{d-1+\ell-1}{\ell-1}-1\right) + 2\cdot n^\ell\cdot\binom{d-1+\ell}{\ell}-1 \\
	&\leq 2\cdot n^\ell\cdot\left(\binom{d-1+\ell}{\ell-1}+\binom{d-1+\ell}{\ell}\right)-1=2\cdot n^\ell\cdot\binom{d+\ell}{\ell}-1.
\end{align*}

We now apply the inequality $\binom{k}{\ell}\leq{\left(\frac{ek}{\ell}\right)}^\ell$ for $k=d+\ell$, obtaining
\begin{align*}
	R(d,\ell) & \leq 2\cdot n^\ell\cdot{\left(\frac{e\cdot(d+\ell)}{\ell}\right)}^\ell
	=2\cdot 2^{\ell\left(\log n+\log e + \log \left(1+\frac{d}{\ell}\right)\right)}.
\end{align*}

Recall that while solving a parity game with $n$ vertices, we start with $p_\Even=p_\Odd=n$, that is, with $\ell=2\cdot\lfloor\log n\rfloor$.
One can check that the above function is increasing, meaning that we can replace $\ell$ by $2\cdot\log n$.
We obtain that
\begin{align*}
	R(d,2\cdot\lfloor\log n\rfloor)&\leq 2\cdot 2^{2\cdot(\log n)\cdot\left(\log n+\log e + \log \left(1+\frac{d}{2\cdot\log n}\right)\right)}
	=2\cdot n^{2\cdot\left(\log n+\log e + \log \left(1+\frac{d}{2\cdot\log n}\right)\right)}.
\end{align*}
This time, to each call to the procedure we associate the cost of the two attractor constructions, preceding and following this call.
In consequence, the cost allocated to each call is again linear in the number of edges.
Estimating $\log e\leq 1.45$, we obtain the following theorem:

\begin{thm}
	Algorithm~\ref{alg:parys} computes winning regions in a given parity game.
	For a parity game with $n$ vertices, $m$ edges, and maximal priority $d$, its memory usage is in $\Oo(m+n\cdot\log n)$, and its running time is in
	\begin{align*}
		\Oo\!\left(m\cdot n^{2.9 + 2\cdot\log n + 2\cdot\log \left(1+\frac{d}{2\cdot\log n}\right)}\right)=n^{\Oo\left(\log n + \log \left(1+\frac{d}{\log n}\right)\right)}
		=n^{\Oo(\log n)}.
	\end{align*}
\end{thm}

The last equality above holds under the assumption $d\leq n$.
This is a reasonable assumption: clearly there are at most $d$ different priorities, and they can be renumbered so that only consecutive numbers, from $1$ to some upper bound, are used.

The main difference between the complexity here and in Section~\ref{sec:liverpool} is that here we have $\log n$ in the exponent.
This means that the algorithm does not speed up in the common case where the number of priorities is significantly smaller than the number of vertices.

\section{Optimisations}\label{sec:optimisations}

In this section we present some optimisations that can be applied to Algorithms~\ref{alg:simple} and~\ref{alg:parys}.
They do not improve the theoretical complexity, but they may (and actually they do) speed up the algorithm in practice, on some inputs.

\begin{enumerate}
\item	After the \textbf{repeat} loop in Algorithm~\ref{alg:parys} we can check whether $|H|\leq\lfloor p_\Odd/2\rfloor$,
	and if this is the case, we can return $G$ immediately,	without executing Lines~\ref{lp16}-\ref{lp23}.
	Indeed, the induction hypothesis guarantees that $W_\Odd$, calculated as $\solve_\Odd(H,d-1,\lfloor p_\Odd/2\rfloor,p_\Even)$,
	contains all Odd's dominions in $H$ of size up to $\lfloor p_\Odd/2\rfloor$.
	When $|H|\leq\lfloor p_\Odd/2\rfloor$, then these are actually all Odd's dominions in $H$, and it makes no sense to call $\solve_\Odd$ again with precision $p_\Odd$.

\item	Likewise, in Algorithm~\ref{alg:simple}, if $|G|\leq\lfloor p_\Odd/2\rfloor$, then the first call to $\solve_\Even$ handles already all Odd's dominions in $G$;
	we can return $G_1$ immediately after Line~\ref{l4}.

\item	To strengthen the effect of testing whether $|G|\leq \lfloor p_\Odd/2\rfloor$ or $|H|\leq\lfloor p_\Odd/2\rfloor$ in Algorithms~\ref{alg:simple} and~\ref{alg:parys}, respectively,
	the initial call can be made with $p_\Even = p_\Odd = 2^{\lfloor \log |G|\rfloor + 1} - 1$ instead of $p_\Even = p_\Odd = |G|$.
	The former number is slightly greater, but the depth of the recursion remains unchanged.

\item	In Algorithm~\ref{alg:parys}, if during the recursive evaluation of the last call to $\solve_\Odd$ in the \textbf{repeat} loop,
	the condition ``$p_\Odd\leq 1$'' (Line~\ref{lp1} of $\solve_\Even$) was never true, then, again, we can return $G$ immediately after this loop, without executing Lines~\ref{lp16}-\ref{lp23}.
	Indeed, in such a case, the call $\solve_\Odd(H,d-1,\lfloor p_\Odd/2\rfloor,p_\Even)$ is equivalent to a call with any greater value (e.g., $|H|$) in place of $\lfloor p_\Odd/2\rfloor$,
	so the the returned set contains all Odd's dominions in $H$, and there is no need to call $\solve_\Odd$ with precision $p_\Odd$.

	The check can be implemented by returning and using a ``dirty'' flag that marks whether the ``$p_\Odd \leq 1$'' condition has been true in any of the sub-routines.

\item\label{opti_5}
	As in implementations of standard Zielonka's algorithm, the conditions in the \textbf{repeat} and \textbf{while} loops in Algorithm~\ref{alg:parys} can be changed
	from ``$W_\Odd = \emptyset$'' to $W_\Odd$ being invariant under $\Odd$-attractor
	(i.e., to ``$\atr_\Odd(W_\Odd,G)=W_\Odd$'' for the value of $G$ before it is updated in Lines~\ref{lp11},~\ref{lp17}, or~\ref{lp22}).
	This requires to re-calculate $H$ after the \textbf{repeat} loop, before Line~\ref{lp16}.
	While this is clearly compatible with the optimisations above, it complicates the correctness argument,
	namely the part where we follow what the algorithm does to a small (size up to $p_\Odd$) Odd's dominion $D$.

	For the basic version of the algorithm, at every step we were considering a nonempty Odd's dominion $C$ in the subgame $D\cap G$, not containing vertices of priority $d$.
	Now, instead of choosing such a dominion $C$ arbitrarily, we should take the greatest one (it exists, because the union of two dominions is a dominion).
	This maximality implies that if $C$ is strictly smaller than $D\cap G$, then $\atr_\Odd(C,D\cap G)$ contains a vertex of the maximal priority $d$.
	Indeed, note that $\atr_\Odd(C,D\cap G)$ is itself an Odd's dominion in $D\cap G$;
	thus it either equals $C$ or contains a vertex of priority $d$.
	In the latter case we are done.
	In the former case, there is a nonempty dominion $C'$ in $(D\cap G)\setminus C$, not containing vertices of priority $d$;
	the union $C'\cup C$ is an Odd's dominion in $D\cap G$, contradicting the maximality of $C$.

	Having this, we observe as previously that if $|C|\leq\lfloor p_\Odd/2\rfloor$,
	then $C$ is included in $W_\Odd$ and thus removed from $G$.
	Moreover, if $C$ does not cover the whole $D\cap G$, then $\atr_\Odd(C,D\cap G)\subseteq\atr_\Odd(W_\Odd,G)$ contains a vertex of priority $d$;
	this vertex is not included in $G\setminus N_d\supseteq H\supseteq W_\Odd$, so the current loop continues.
	If $|C|>\lfloor p_\Odd/2\rfloor$, then the \textbf{repeat} loop may finish without removing the whole $C$ from $G$,
	and the remaining part of $C$ is removed from $G$ by the full precision call in Lines~\ref{lp16}--\ref{lp17} (since clearly $|C|\leq p_\Odd$).
	After that, in the \textbf{while} loop, we always have $|C|\leq|D\cap G|\leq\lfloor p_\Odd/2\rfloor$, which, by the above, guarantees removal of all remaining vertices of $D$.

\item	With a similar argument, the third call of Algorithm~\ref{alg:simple} (Line~\ref{l9}) can be skipped if the $\Odd$-attractor of $W_\Odd$ in $G_1$ equals $W_\Odd$.
\end{enumerate}

\section{Evaluation}

Our evaluation focuses on the performance of the presented algorithms on random parity games on the one hand and on so-called worst-case families on the other hand.
In a nutshell, the results are without great surprises. On random games, which usually don't present much of a challenge for Zielonka's algorithm, its quasipolynomial-time variations are in general orders of magnitude slower. The exception is games that can be solved in very few iterations by Zielonka, such as random high-degree games, which seem to be solvable by the Warsaw quasipolynomial-time algorithm in just as few iteration. Furthermore, the Liverpool quasipolynomial-time algorithm, which has the best worst-case complexity, reliably requires more iterations than the Warsaw version, which has a higher worst-case complexity.
On the families designed to trigger worst-case complexity in Zielonka's algorithm~\cite{Fri11,benerecetti2019robust,Gaz2016}, the story is different. By game-size 54, both quasipolynomial-time algorithms require fewer iterations than the standard recursive algorithm, and the difference in performance from then on grows at an exponential rate. On these families of worst-case games, as the theory predicts, the algorithm with better worst-case behaviour indeed requires fewer iterations.

\subsection{Implementation}

We have implemented the two algorithms presented in the paper, using all the optimisations from Section~\ref{sec:optimisations}.
For comparison, we have also implemented standard Zielonka's algorithm, where we apply Optimisation~\ref{opti_5} from Section~\ref{sec:optimisations}
(i.e., we use ``$\atr_\Odd(W_\Odd,G)=W_\Odd$'' instead of ``$W_\Odd = \emptyset$'' for the condition of a loop).
Our implementation involves the Oink framework~\cite{oink} to read files with games and to confirm correctness of our solutions,
but apart from that the implementation is independent.

To compare performance of the three algorithms, we use the number of iterations.
Because the most costly operation in the procedure is the computation of an attractor,
as a single iteration we count the fragment of code where an attractor for one player is computed, a subprocedure is called, and an attractor for the other player is computed
(the first computation of an attractor is not present in Lines~\ref{lp16}--\ref{lp17} of Algorithm 2, but anyway we count these two lines as an iteration).
Note, however, that iterations are exactly the places where a subprocedure is called,
so the number of iterations is actually equal to the number of recursive calls, minus one.

With the number of iterations as a convenient running time measure,
our evaluation becomes independent from low-level optimisations or from a choice of hardware.

The implementations, together with detailed results of the evaluation, are available at \href{https://github.com/pparys/qpt-parity}{https://github.com/pparys/qpt-parity}.

\begin{samepage}
\subsection{Benchmarks}

\subsubsection*{Constructed parity games}
Friedmann~\cite{Fri2011}\end{samepage} was the first to show an exponential lower bound to the running time of Zielonka's recursive algorithm. Since then, several families designed to trigger worst-case behaviour in recursive algorithms have been proposed: some are more robust~\cite{benerecetti2019robust}, in the sense that they also provide lower bounds for many variations and optimisations of the standard algorithm; others provide stronger lower bounds~\cite{Gaz2016}. Since here we do not study optimisation strategies, we choose to present the performance of our algorithm on Gazda's family of games, which provide the strongest lower bound we are aware of to date, namely $\Omega(2^{\frac{n}{3}})$~\cite{Gaz2016}. From our observations, this is representative of the algorithms' performance on any of the mentioned families. We have run the three algorithms on games of size up to 231, above which point all the three algorithms time out at 15 minutes.

\subsubsection*{Random parity games}
We use the PG-solver tools~\cite{PG-solver} to generate random parity games. We study the algorithms both on games with high vertex degree and with low vertex degree.
Overall, we have run our algorithms on 320 random games of sizes ranging from 100 to 2000; the number of priorities always corresponds to the number of vertices. For random games, the game size is a poor predictor of how many iterations are needed---for instance, in our sample, Zielonka's algorithm solves more random games with high degree of size 2000 than of size 200 in two iterations---we therefore don't report the sizes of the sample games in our graphs.

\subsection{Results}

We compare the implementations of the three algorithms, with all the optimisations from Section~\ref{sec:optimisations} included.
We use the number of iterations required by each algorithm as a proxy for their performance.

\subsubsection*{Constructed parity games}
\begin{figure}
\includegraphics[width=0.95\textwidth]{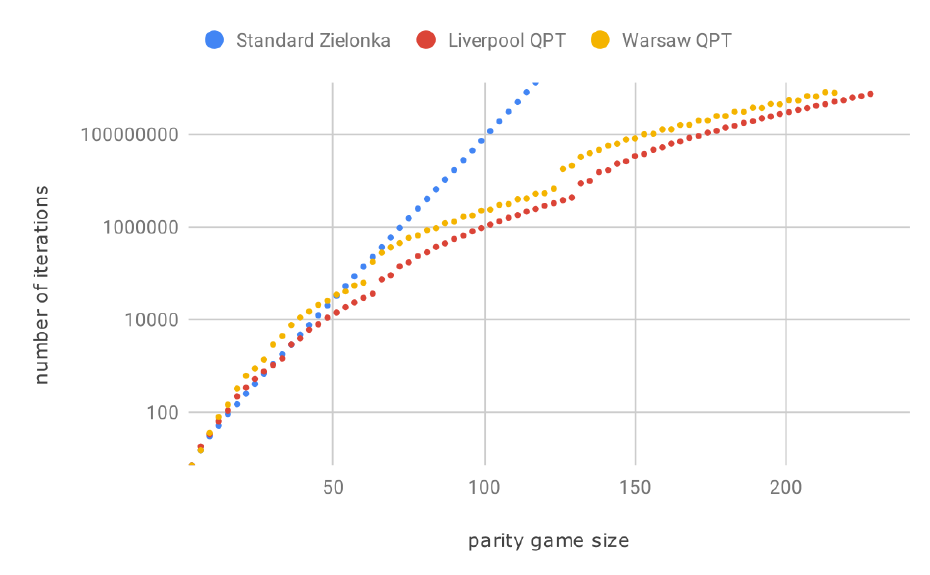}
\caption{Number of iterations needed by each algorithm on Gazda's family of games}\label{fig-worst}
\end{figure}
As shown in Figure~\ref{fig-worst}, on parity games of Gazda's family, the performance of both quasipolynomial-time versions overtake the standard version of Zielonka's algorithm around size 50.
From there on, the performance gap grows exponentially.
The Liverpool version of the algorithm, which enjoys better worst-case complexity, outperforms the Warsaw version consistently.
Note that the logarithmic scale hides the growth of this performance gap.
The number of iterations required by both quasipolynomial-time algorithms seems to jump abruptly at parity games of size 30, 63 and 126
(these sizes are around the powers of $2$, where the search depth of the algorithms change); this behaviour is persistent throughout the different worst-case families.

\subsubsection*{Random parity games}
\begin{figure}
\includegraphics[width=0.95\textwidth]{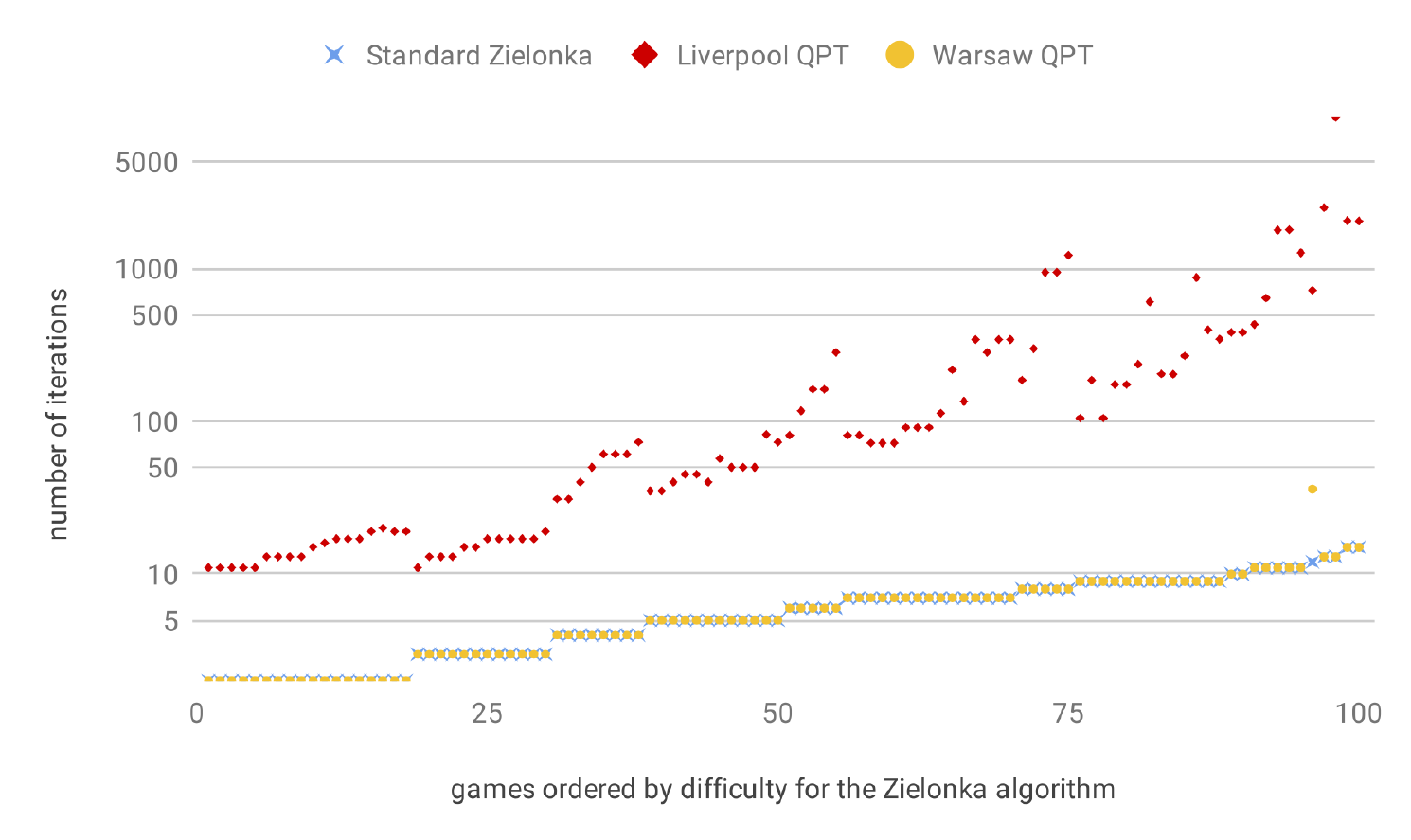}
\caption{Number of iterations needed by each algorithm on randomly generated games of high degree with between 100 and 2000 vertices}\label{fig:high-degree}
\end{figure}
\begin{figure}
\includegraphics[width=0.95\textwidth]{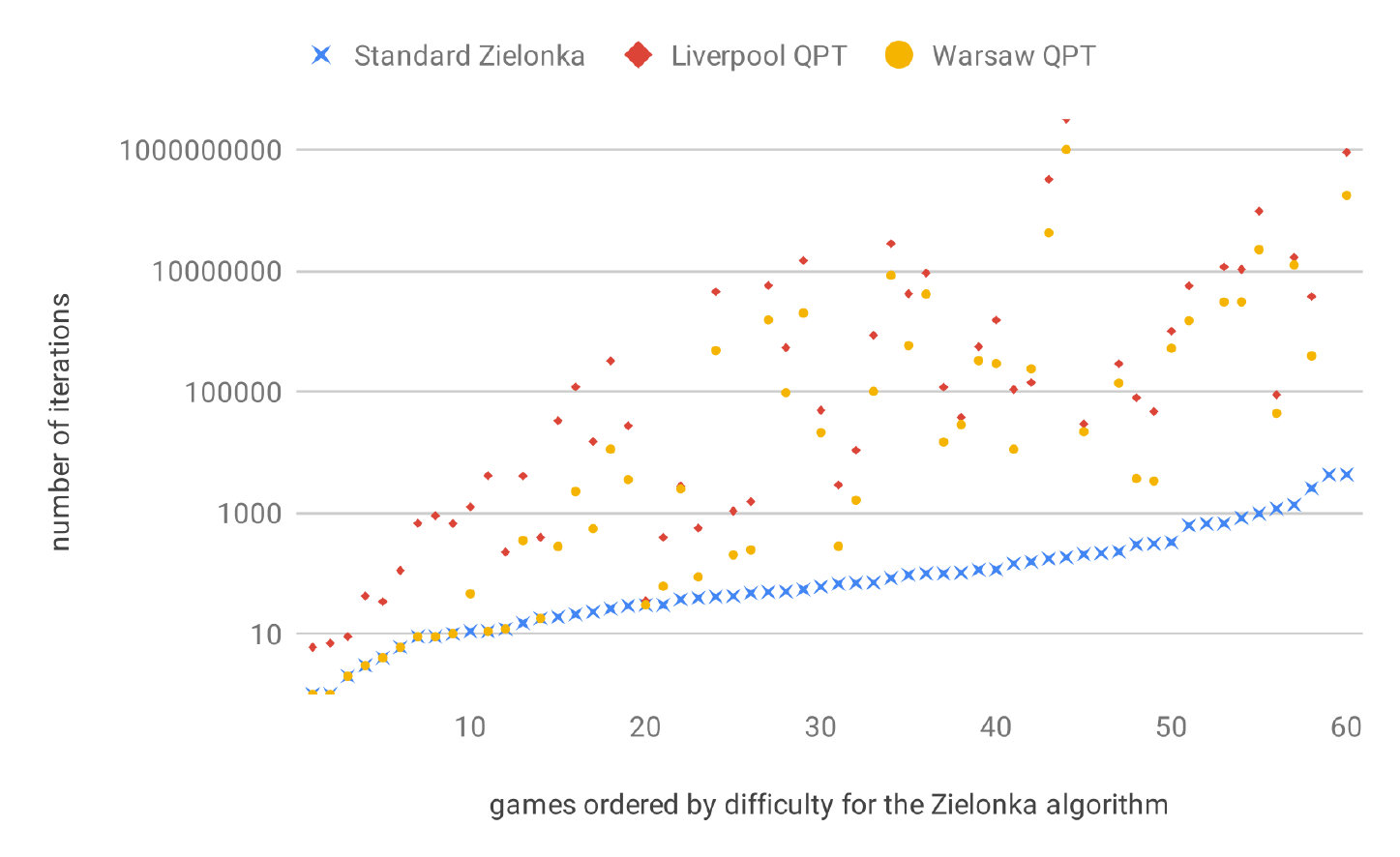}
\caption{Number of iterations needed by each algorithm on randomly generated games of low degree with between 100 and 500 vertices}\label{fig:low-degree}
\end{figure}
The size of randomly generated games does not generally predict the performance of the algorithms.
On games of high degree, the algorithms agree on the relative difficulty of games; see Figure~\ref{fig:high-degree}.
Both Zielonka's algorithm and the Warsaw quasipolynomial-time version consistently solve these games in under 15 iterations,
while the Liverpool version needs one to two orders of magnitude more iterations.
The reason why the Warsaw version behaves like Zielonka's algorithm in this case is that differences only start if the recursion depth of the algorithms exceeds the logarithm of the game size;
here the Warsaw version never runs out of bound, its call structure is the same as Zielonka's.
On random games of low degree, see Figure~\ref{fig:low-degree}, the two quasipolynomial-time versions only broadly tend to agree with Zielonka's algorithm on the relative difficulty of games. Again, the Warsaw quasipolynomial-time version seems to match Zielonka's algorithm on the easiest games, namely those that Zielonka's algorithm solves in under 15 iterations. However, in general both the Warsaw and Liverpool version require several orders of magnitude more iterations and the Warsaw version consistently outperforms the Liverpool version throughout.

\section{Conclusion}

We have presented two quasipolynomial-time algorithms solving parity games, working recursively.
The main advantage of our recursive approach over previous quasipolynomial-time algorithms solving parity games lies in its simplicity:
we perform a small modification of the straightforward Zielonka's algorithm.

Our original hope was that, in practice, the quasipolynomial-time algorithms can be as fast as standard Zielonka's algorithm on typical inputs,
while being faster (quasipolynomial instead of exponential) on worst-case inputs.
This turned out to be only partially true: although our algorithms are not dramatically slower, a significant slowdown is visible
(except for the Warsaw variant, in case of a really small number of iterations).

Czerwiński et al.~\cite{sep-lower-bound} argue that previous quasipolynomial-time algorithms solving parity games~\cite{CJKLS17,JL17,Leh18,FJKSSW19} are instances of a so-called separator approach.
Then, they prove that every algorithm accomplishing this approach has to follow a structure of a universal tree (a tree into which every tree of a given size can be embedded),
and they show a quasipolynomial lower bound for the size of such a tree---hence also for the running time of the algorithm.
Our algorithms, on the one hand, are of a different style; they cannot be seen as instances of the separator approach, hence the lower bound does not apply to them.
On the other hand, however, one can see that the trees of recursive calls in our two algorithms follow two particular constructions of universal trees.
Furthermore, Jurdziński and Morvan~\cite{JM20} generalise our algorithms to a generic algorithm parameterised by a universal tree
(any universal tree gives a correct algorithm, while our two variants are obtained by using particular universal trees).
Arnold, Niwiński, and Parys~\cite{mu-calculus} show a similar generalisation, from a slightly different perspective of fixed-point evaluation;
they also prove that universal trees are persistently present in their approach, resulting in a quasipolynomial lower bound for the number of recursive calls.

Finally, let us mention a recent result of Jurdziński, Morvan, Ohlmann, and Thejaswini~\cite{symmetric}:
they design an algorithm solving parity games that is symmetric (like our recursive algorithms),
but simultaneously works in time proportional to the size of one universal tree, not to the size of a product of two such trees
(thus the logarithm in the exponent is not multiplicated by $2$, unlike for our Liverpool variant).

\bibliographystyle{alphaurl}
\bibliography{refs}%
\label{sec:biblio}

\end{document}